\documentclass{article}
\usepackage{biblatex}
\usepackage[utf8]{inputenc}
\usepackage[english]{babel}
\usepackage[margin=1in]{geometry}
\usepackage{amsmath,amssymb,amsfonts,amsthm,mathtools}
\usepackage{xcolor}
\usepackage{tikz}
\usetikzlibrary{fit}
\usepackage{quantikz}
\usepackage[inline]{enumitem}

\usepackage{hyperref}
\usepackage[nameinlink]{cleveref}

\hypersetup{colorlinks=true,linkcolor=blue,filecolor=magenta,urlcolor=cyan}
\newtheorem{theorem}{Theorem}[section]

\newtheorem{lemma}[theorem]{Lemma}
\theoremstyle{remark}
\newtheorem*{remark}{Remark}
\theoremstyle{definition}

\newtheorem{prop}{Proposition}
\newtheorem{conjecture}{Conjecture}
\newcommand{\R}{\mathbb{R}}
\newcommand{\C}{\mathbb{C}}

\title{A remark on the Brown-Susskind conjecture  }
\author{Ranee Brylinski\thanks{She passed away before this paper was submitted} and Jean-Luc Brylinski}

\date{September 2026}

\begin{document}
\maketitle

\begin{abstract}
In the spirit of the Brown-Susskind conjecture, proven in \cite{haf} and \cite{li}, we study the growth of the dimension of the set of $n$-qubit unitaries which can be obtained as a product of a fixed number of $2$-qubit gates from a fixed set of pairs of qubits. We show that this dimension increases strictly when one more pair is added,  for at least some choice of the $2$-qubit pairs at each step.
\end{abstract}

\section{Introduction}
Brown and Susskind \cite{sus-br} conjectured that a local random circuit’s quantum complexity grows linearly in the number of gates
until reaching a value exponential in the system size. The conjecture was proven in \cite{haf}, and two other proofs were given by Zhi Li \cite{li}. In this paper we show that the complexity will strictly increase at each step of the circuit for at least one choice of the $2$-qubits used in the next step of the circuit.

We denote by $U(2^n)$ the group of unitary gatess on $n$ qubits.
We will work in the context of a fixed set $S$ of unordered qubit pairs $\alpha=ij$ and the corresponding $U(4)$ subgroups of $U(2^n)$, called $G_{\alpha}$, of $2$-qubit gates for these two qubits. From this set $S$ we construct a graph $\Gamma$ whose nodes are the $n$ qubits and whose edges are the pairs $(ij)$ in $S$.  We then have

\begin{prop}The subgroups $G_{\alpha}$ where the $\alpha$ range over $S$ generate $U(2^n)$ if and only if the graph $\Gamma$ is connected.

\end{prop}

\begin {proof} First suppose the graph $\Gamma$ is connected. Let $H$ be the closure of the subgroup of $U(2^n)$ generated by the $G_{\alpha}$ where the $\alpha$'s range over $S$. As a closed subgroup of a compact Lie group, $H$ is a Lie subgroup so it has a  Lie algebra $\frak h$, which is a Lie subalgebra of the Lie algebra $\frak u(2^n)$ of $U(2^n)$. T

We will show that any Pauli string on $n$ qubits belongs to $\frak h$. Here by Pauli string we mean a tensor product $P_1\otimes\cdots\otimes P_n$ where each $P_j$ is a Pauli operator (possibly $I,X,Y,Z$) acting on the $j$-th qubit. If we have Pauli operators $P_1,\dots,P_k$ supported on $k$ qubits, we will denote by $P_1\otimes\cdots\otimes P_n$ the Pauli string on $n$ qubits where $I$ is inserted for the missing qubits.

Given a path $\tau=i_1\cdots i_{D+1}$ of length $D+1$  in $\Gamma$, a Pauli string 
$P_1\otimes P_{D+1}$ is said to be supported on 
 $\tau$ if each $P_j$ belongs to the subgroup $G_{i_ji_{j+1}}$.

We will show by induction on $D$ that any such Pauli string belongs to $\frak h$.

To start the induction, if $D=1$, $\alpha=(i_1i_2)$ belongs to $S$ and so $P_{i_1}\otimes P_{i_2}$ lies in $\frak h$.

For the induction step, we will for simplicity write $P_1$ for $P_{i_1}$ and $P_2$ for $P_{D+1}$. We may assume $P_1\neq I$, otherwise our Pauli string is supported on a path of smaller length $D$, so it belongs to $\frak h$ by induction. Then we can write $P_1$ as a bracket $P_1=1/2\, [Q,R]$ for some Paulis $Q,R$ bn the qubit $i_1$.
We will write our length $D+1$ Pauli string as $P_1\otimes W\otimes P_{D+1}$ where the Pauli string $W$ lives on qubits $2$ thru $D$. If $W=I$, then we have

\begin{equation}
P_1\otimes I\otimes P_2=\frac{12}\,[Q\otimes V\otimes I,\,R\otimes V\otimes P_2]
\end{equation}
 for any non-trivial Pauli string $V$ on qubits $2$ thru $D$. 
Now $Q\otimes V\otimes I$ belongs to $\frak h$ by the inductive hypothesis, since it is supported on a path of length $D$. Writing $V=\frac{1}{2}\, [A,B]$ for some Pauli strings $A,B$, we have
\begin{equation*}
R\otimes V\otimes P_2 \,=\frac{1}{2}\,[\,R\otimes A\otimes I,\ I\otimes B\otimes P_2\,].
\end{equation*}
a bracket of two Pauli strings which are themselves in $\frak h$ by induction. Hence $P_1\otimes I\otimes P_2$
 belongs to $\frak h$. 

On the other hand, if $W\neq I$, we have $W=\frac{1}{2}\, [U,V]$  for some Pauli strings $U,V$ and we obtain

\begin{equation}
P_1\otimes W\otimes P_2=\frac{1}{2}\, [P_1\otimes U\otimes I,I\otimes V\otimes P_2]
\end{equation}
This is a bracket of two operators which both belong to $\frak h$ by induction.

Thus $\frak h$ contains the Pauli string basis of $\frak u(2^n)$, so it is equal to $\frak u(2^n)$. Hence the closed subgroup subgroup $H$ is equal to $U(2^n)$. 

Then by \cite{bry, Lemma 4.2}, the closed subgroups $G_{\alpha}$ for $\alpha\in S$ generate $U(2^n)$ as an  abstract group, i.e. any operator in $U(2^n)$  is a finite product of operators which belong to one of the subgroups.

\vskip .13 in

On the other hand, if $\Gamma$ is not connected, then all the $G_{\alpha}$ for $\alpha\in S$ belong to a product of several smaller unitary groups corresponding to the connected components of $\Gamma$, hence all these subgroups put together cannot generate $U(2^n)$.
\end{proof}

We consider a sequence $\tau=\alpha_1,\dots,\alpha_k$ of qubit pairs and the corresponding construction map
\begin{equation*}
F_k: G_{\alpha_1}\times\cdots\times G_{\alpha_k}\to U(2^n).
\end{equation*}

\begin{equation}
F_k(g_1,\dots,g_k)=g_1\cdots g_k
\end{equation}
We denote the image of $F_k$ by $\mathcal U_k$, similarly to \cite{li}.

In contrast to \cite{haf} and \cite{li}. we do not choose an architecture, which means a sequence of qubit pairs which keeps repeating, like the staircase architecture or the wheel spoke architecture. Instead, we let the sequence of $2$-qubit pairs evolve freely within our alphabet $S$ of qubit pairs. Our main result is

\begin{theorem} \label{th} 
Given a  sequence $\tau=\alpha_1,\cdots,\alpha_k$ of qubit pairs for which $\mathcal U_k$ does not have maximal dimension, there exists a choice of another $\alpha_{k+1}$ in the alphabet $S$such that $dim(\mathcal U_{k+1})> dim(\mathcal U_k)$.

\end{theorem}

We point out that we do not at all know how to choose the next qubit in the sequence in order to obtain this dimension increase. It may be that a good choice would be a qubit pair for which the new edge in the graph would be far away from recently used edges, but this is just speculation at this point. This uncertainty does not occur in the works \cite{haf} \cite{l}, which ioop recursively thru the $2$-qubit gates in a chosen architecture.
\vskip .1 in

The notion of dimension for $\mathcal U_k$ needs some explanation, which we give in the next section.

\section{Proof of  \cref{th}}

First let us explain what kind of objects the subsets $\mathcal U_k$ of $U(2^n)$ are.  The subset $\mathcal U_k$ of $U(2^n)$is the image of the construction map
$F_k:G_{\alpha_1}\times\cdots\times G_{\alpha_k}\to U(2^n)$. Now $U(2^n)$ and the $G_{\alpha}$'s are real-algebraic varieties, as they are subsets of the real vector space of complex $2^n\times 2^n$ matrices defined by algebraic equations. Here we note that every complex vector space is certainly a real vector space, with a doubling of the dimension. To wit, $U(2^n)$ is defined inside the vector space of $2^n\times 2^n$ complex matrices by the matrix equation $UU^{\dagger}=I$, and the entries of $UU^{\dagger}$ are polynomials in the real and imaginary parts of the entries of $U$. Similarly each $G_{\alpha_l}$ is defined by real-algebraic equations inside the vector space of $4\times 4$ matrices, namely by the conditions of commuting with each one-qubit Pauli string on a qubit which does not belong to the edge $\alpha_l$.
 
The construction map $F_k$ is thus a map from a real-algebraic variety to another one. Its image $\mathcal U_k$ is likely not to be a real-algebraic subvariety; after all, the image of the squaring map $x\to x^2$ from $\R$ to $\R$ is the set
$\R^+$ of non-negative real numbeers, which is not an algabraic subvariety:  it is defined by the polynomial inequality $x\ge 0$.

Now we have the following classical finiteness theorem, which is used extensively for decidability questions in computer science and in complexity theory. We state here the strong version of the theorem.

\begin{theorem}[Tarski--Seidenberg theorem]\label{TS}
The image $F(V)$ of a polynomial map $F:V\to W$ between real-algebraic varieties is a semialgebraic subset of $W$.
\end{theorem}

Now what are semialgebraic subsets of $\R^n$? They are finite unions of basic semialgebraic subsets, which are cut out in $\R^n$ by finitely many polynomial equations $P(x_1,\dots,x_n)=0$ and finitely many polynomial inequalities
$Q(x_1,\cdots,x_n)>0$. For inatance, semialgebraic subsets of $\R$ are exactly the finite unions of points and of intervals which may be open  or closed at either end. 

\vskip .12 in

It turns out that any algebraic set $V\subset \R^n$ has a \emph{dimension}, say $D$, which may be given three equivalent definitions:

\vskip .12 in

(1) $D$ is the largest dimension of an open ball that can be embedded smoothly inside $V$

(2) $V$ contains an open subest $U$ of $V$ which is dense in $V$ and has dimension $D$.

(3) the Zariski closure $\bar{V}$ of $V$ has dimension $D$ in the sense of real-algebraic geometry, that is the field of fractions of the algebra of polynomial functions on $V$ has transcendance deeree $D$.

\vskip .12 in

Now the Zariski closure $\bar{V}$ of $V$ is the algebraic subvariety of $\R^n$ cut out by all the polynomial equations
$P(x_1,\dots,x_n)=0$ which hold on $V$.

For instance, consider the half-cone $C$ in $\R^3$ defined by the conditions
$x^2+y^2-z^2=0,\, z\ge 0$. $C$ is a semialgebraic set, all the polynomials vanishing on $C$ are multiples of $x^2+y^2-z^2$, hence the Zariski closure $\bar{C}$ is the full cone cut out by the equation $x^2+y^2-z^2=0$.

\vskip .12 in

Now we are in position to prove Theorem \cref{th}.

\begin{proof} Let $D$ be the dimension of $\mathcal U_k$, so that $D<\dim(U(2^n))$. Assume no $G_{\alpha_{k+1}}$ exists for which $\mathcal U_{k+1}$ has dimension $>D$. Then we have the inclusion $\overline{\mathcal U_k}\subseteq \overline{\mathcal U_{k+1}}$ between two real-algebraic subvarieties of the same dimension. Now we have

\begin{lemma} The real-algebraic variety $ \overline{\mathcal U_k}$ is irreducible, i.e. if  $ \overline{\mathcal U_k}$  is the union of two Zariski closed subsets $SA$ and $B$, one of $A$ and $B$ is equal to $\overline{\mathcal U_k}$.

\end{lemma}

Here a Zariski closed subset is a subset which is cut out by finitely many polynomial equations. These are the closed subsets for a topology called the Zariski topology. Any  map betwen two real-algebraic varieties is continuous for the Zariski topologies of the domain and the target.
\vskip 1 in
To prove the lemma, first recall that $\mathcal U_k$ is the image of the polynomial map
\begin{equation*}
F_k: G_{\alpha_1}\times\cdots\times G_{\alpha_k}\to U(2^n).
\end{equation*}
Write $X=G_{\alpha_1}\times\cdots\times G_{\alpha_k}$ for short. We then have

\begin{align*}
X&=F_k^{-1}(F_k(X))\\
&=F_k^{-1}(\overline{F_k(X)})\\
&=F_k^{-1}(A)\cup F_k^{-1}(B)
\end{align*}

Since $X$ is a product of irreducible algebraic varieties, it is itself irreducible. Its subsets $F_k^{-1}(A)$ and $F_k^{-1}(B)$ are Zariski closed, as the map $F_k$ is Zariski continuous. Hence either $X=F_k^{-1}(A)$ or $X=F_k^{-1}(B)$. In the first case we have $F_k(X)\subseteq A$, hence $\overline{\mathcal U_k}=\overline{F_k(X)}=A$; the second case is similar.
\vskip .13 in

Now, going back to our inclusion $\overline{\mathcal U_k} \subseteq \overline{\mathcal U_{k+1}}$, both are irreducible algebraic varieties of dimension $D$, and so these two subvarieties must be equal. Hence $\overline{\mathcal U_k}=\overline{\mathcal U_{k+1}}$.

Now recalling that $\mathcal U_{k+1}$ is the product $G_1\cdots G_{k+1}$ inside $U(2^n)$, we have $\mathcal U_{k+1}=\mathcal U_k\,G_{k+1}$. When we saturate the Zariski closure $\overline{\mathcal U_k}$ under the action of $G_{k+1}$ by right multiplication, we observe that
\begin{equation}
\overline{\mathcal U_k}\,G_{k+1}\subseteq \overline{\mathcal U_k\,G_{k+1}}=\overline{\mathcal U_{k+1}}=\overline{\mathcal U_k}.
\end{equation}

The upshot is that $\overline{\mathcal U_k}$ is carried into itself by the right action of $G_{k+1}$. But as we allow 
$G_{k+1}$ to vary among all the subgroups $G_{\alpha}$ where $\alpha$ ranges rhe alphabet $S$, and these subgroups generate $U(2^n)$, it ensues that $\overline{\mathcal U_k}$ is carried into itself by the right action of  the whole of $U(2^n)$. Bur then $\overline{\mathcal U_k}$ is equal to $U(2^n)$, a contradiction.

\end{proof}.

\begin{remark} The inclusion $\overline{\mathcal U_k}\times G_{k+1}\subseteq \overline{\mathcal U_k\times G_{k+1}}$ alreafdy is used in \cite{li}, with different notations.

\end{remark} 

\section{The case of two one-parameter subgroups}
We consider two one-parameter subgroups $T_1=\{\exp(itH_1)\}$, $T_2=\{\exp(itH_2)\}$ inside $U(2^n)$, where the $H_i$'s are Hamiltonians, i.e. Hermitian operators on $\C^{2^n}$. Then we consider the subsets $\mathcal U_k=T_1T_2\cdots$ consisting of $k$-fold products taken from these two subgroups in order. Physically this is a time-dependent Hamiltonian evolution which jumps from one Hamiltonian to the other at random times. We believe the following conjecture is due to Seth Lloyd.

\begin{conjecture}
Assume $iH_1$ and $iH_2$ generate the Lie algebra of skew-Hermitian operators on $(\C^2)^{\otimes n}$ and that the one-parameter subgroups $T_1$ and $T_2$ are periodic. Then for $k\le 2^{2^n}$ the dimension of $\mathcal U_k$ is exactly equal to $k$.
\end{conjecture}

\begin{proof}
Being periodic, the one-parameter subgroups are closed real-algebraic subvarieties of $U(2^n)$. Then the methods of this paper show that the dimension of $\mathcal U_k$ increases by at least $1$ from $k$ to $k+1$, noting that at each step one is forced to change from one-parameter subgroup to the other. Of course the dimension cannot increase by more than $1$, so the increase is exactly $1$.
\end{proof}

In case the two one-parameter subgroups are not both periodic, one of them  no longer is a real-algebraic subvariety. It seems possible that one could tackle this case by viewing them  as subanalytic spaces in the sense of Heisuke Hironaka, but we have not pursued that.

\printbibliography

@article{haf,
  title     = {Linear growth of quantum circuit complexity},
  author    = {Haferkamp, Jonas and Faist, Philippe and Kothakonda, Naga B. T. and Eisert, Jens and Yunger Halpern, Nicole},
  journal   = {Nature Physics},
  publisher = {Springer Science and Business Media LLC},
  volume    = {18},
  number    = {5},
  pages     = {528--532},
  year      = {2022},
  month     = {mar},
  issn      = {1745-2481},
  doi       = {10.1038/s41567-022-01539-6},
  url       = {https://doi.org/10.1038/s41567-022-01539-6}
}

@article{sus-br,
  author  = {Brown, Adam and Susskind, Leonard},
  title   = {Second law of quantum complexity},
  journal = {Phys. Rev. D},
  volume  = {97},
  pages   = {086015},
  year    = {2018}
}

@misc{bry,
      title={Universal quantum gates}, 
      author={Jean-Luc Brylinski and Ranee Brylinski},
      year={2001},
      eprint={quant-ph/0108062},
      archivePrefix={arXiv},
      primaryClass={quant-ph},
      url={https://arxiv.org/abs/quant-ph/0108062}, 
}

@misc{li,
      title={Short Proofs of Linear Growth of Quantum Circuit Complexity}, 
      author={Zhi Li},
      year={2022},
      eprint={2205.05668},
      archivePrefix={arXiv},
      primaryClass={quant-ph},
      url={https://arxiv.org/abs/2205.05668}, 
}
\end{document}